\documentclass[10pt,twocolumn]{article}
\usepackage[T1]{fontenc}
\usepackage[utf8]{inputenc}
\usepackage{mathptmx}
\usepackage[letterpaper,margin=0.75in,columnsep=20pt]{geometry}
\usepackage{microtype}
\usepackage{cite}
\usepackage{amsmath,amssymb,amsfonts}
\usepackage{algorithmic}
\usepackage{graphicx}
\usepackage{textcomp}
\usepackage{booktabs}
\usepackage{bm}
\usepackage{fancyhdr}
\fancypagestyle{plain}{\fancyhf{}\fancyhead[C]{\small\itshape Preprint. Submitted to IEEE Access.}\fancyfoot[C]{\small\thepage}}

\def\BibTeX{{\rm B\kern-.05em{\sc i\kern-.025em b}\kern-.08em
    T\kern-.1667em\lower.7ex\hbox{E}\kern-.125emX}}

\newtheorem{theorem}{Theorem}
\newtheorem{proposition}{Proposition}
\newtheorem{lemma}{Lemma}

\newtheorem{definition}{Definition}
\newtheorem{remark}{Remark}
\newtheorem{example}{Example}

\begin{document}

\title{Set-Packing and Sequence-Pair QUBOs for the 2D Cutting Stock Problem on Quantum Annealing Hardware}

\author{Miguel S\'anchez-Beato\thanks{Corresponding author: Miguel S\'anchez-Beato (e-mail: miguel.sanchez@itecam.es).}%
\thanks{This work has been submitted to the IEEE for possible publication. Copyright may be transferred without notice, after which this version may no longer be accessible. This work was carried out at the AI and Quantum Technologies Department, ITECAM, Spain.}\,, Raul Martinez, Matilde Osa, Jorge Parra, and Mario Calonge\\[3pt]
\normalsize AI and Quantum Technologies Department, ITECAM, Spain}
\date{}

\maketitle

\begin{abstract}
The two-dimensional Cutting Stock Problem (2D-CSP) is an NP-hard problem with direct economic and environmental impacts on manufacturing and logistics. We encode its fixed-plate variant, with free piece repetition and full non-overlap and containment constraints, as a Quadratic Unconstrained Binary Optimization (QUBO) problem for quantum annealing and compare two formulations from opposite encoding paradigms. The first was a coordinate-based set-packing model with one binary variable per candidate placement. Its variable count grows linearly with plate area and resolution, but its ground state is, by construction, a geometrically feasible maximum-area packing. The second is a coordinate-free sequence-pair model whose variable count is independent of plate resolution and size. We prove that this compactness has a structural limit: no coordinate-free QUBO of bounded interaction degree whose penalties vanish on every geometrically feasible layout can have a geometrically feasible ground state for 2D containment, because containment is a longest-path constraint that bounded-degree penalties cannot enforce on chains longer than their interaction order. We evaluate both formulations under multi-seed simulated annealing, simulated quantum annealing, and an exact integer-programming baseline. Hardware experiments include D-Wave minor embedding, a calibrated direct-QPU sweep, and Leap hybrid solvers in both penalty and constraint-native form, across a six-instance campaign with per-instance calibration. The hybrid solver returns our certificate configuration, tying its energy to thirteen decimal places while overflowing the plate. We claim no quantum speedup. Our contribution is an impossibility result characterizing the limits of compact packing QUBOs, and a practical rule for choosing between the two formulations.
\end{abstract}


\section{Introduction}
\label{sec:introduction}
Minimizing material waste has direct consequences both the economic profitability and environmental footprint of the manufacturing industry. Optimal cutting patterns reduce production costs across the entire supply chain, from raw material acquisition to final product pricing. In this context, the two-dimensional Cutting Stock Problem (2D-CSP) is a foundational combinatorial optimization problem. Given a stock plate of fixed dimensions and a catalogue of demanded piece types, determine an arrangement of pieces on the plate that maximises the occupied area subject to non-overlap and containment constraints.

Exact methods based on integer programming guarantee optimal solutions but scale poorly with the number of pieces and plate dimensions \cite{beasley1985exact, iori2021exact}. Heuristics such as bottom-left placement, first-fit-decreasing, and metaheuristics trade optimality for tractability \cite{baker1980orthogonal, hopper2001empirical}. This tradeoff motivates the search for alternative optimization paradigms. Although classical heuristics are often trapped in local minima or require over-simplified geometric restrictions, Quantum Annealing offers an alternative computing paradigm. By initialising the system in the ground state of a transverse-field driver Hamiltonian and adiabatically interpolating toward the problem Hamiltonian, QA enables tunnelling through energy barriers that classical thermal fluctuations would have to climb over \cite{kadowaki1998,hauke2020}. Although a general quantum speedup over the best classical heuristics remains an open empirical question \cite{ronnow2014,yarkoni2022}, the approach has shown promise on industrial combinatorial problems and merits exploration in the cutting-stock setting.

A QUBO formulation of 2D-CSP must encode three things simultaneously: which pieces are selected, where they sit, and the geometric constraints of non-overlap and containment. Two encoding strategies sit at opposite ends of a design-space spectrum, defined by how coordinates are represented. A coordinate-based encoding introduces one binary variable per candidate placement of a piece on the plate. The resulting model is a weighted set-packing (independent-set) problem in the canonical form introduced by Lucas \cite{lucas2014}, and its variable count grows linearly with the plate area. A coordinate-free encoding encodes only the relative order of pieces through two permutations, recovering physical coordinates via a classical longest-path computation. This approach, pioneered by Murata et al. \cite{murata1996}, brought to Ising form by Terada et al. \cite{terada2018}, and refined by Okada et al. \cite{okada2024}, produces a variable count far smaller than the coordinate-based formulation on instances with a large number of candidate positions. A grid reduction based on the greatest common divisor narrows this gap on the set-packing side, though only when the dimensions share a common factor. The fundamental trade-off between the two strategies is thus clear: variable economy versus guaranteed ground-state feasibility.

We show that this trade-off is not free, and that the reason is structural rather than a matter of tuning. The coordinate-based set-packing model is exact: its ground state is, by construction, a geometrically feasible maximum-area packing. The coordinate-free model is not. Containment of a piece is a longest-path constraint over a chain of pieces, whereas penalties of bounded interaction degree can only see boundedly many pieces at a time. We prove that no coordinate-free QUBO of bounded degree whose penalties vanish on every feasible layout admits a geometrically feasible ground state for 2D containment: a chain of $d+1$ pieces can overflow the plate while every $d$ of them fit, so a degree-$d$ penalty that is exact on feasible layouts must assign the overflowing chain zero penalty as well. A faithful coordinate-free penalty would require interaction terms of order equal to the chain length, which is prohibitive on near-term annealing hardware. Penalties that also fire on feasible layouts can displace particular infeasible states, but only with instance-specific tuning and at the cost of the energy-area correspondence, a degradation we quantify empirically. The compact formulation is therefore adequate only as a sample-and-filter heuristic: sample low-energy configurations, decode their coordinates, discard those that overflow the plate, and keep the best feasible packing. It cannot be used as an exact ground-state encoding.

\textit{Contributions.} This paper makes the following contributions:
(i) a side-by-side QUBO treatment of the fixed-plate 2D Cutting Stock Problem with free piece repetition, covering a coordinate-based set-packing formulation and a coordinate-free sequence-pair formulation, with closed-form variable and coupler counts and a greatest-common-divisor grid-reduction observation that sharpens the comparison between them;
(ii) an impossibility result, namely that no bounded-degree coordinate-free QUBO with penalties that vanish on all feasible layouts can have a geometrically feasible ground state for 2D containment, together with an explicit zero-penalty infeasible configuration of our tuned Hamiltonian that is energy-degenerate with the optimum;
(iii) to our knowledge, the first evaluation of the fixed-plate, free-repetition CSP variant with full non-guillotine non-overlap and containment on real quantum annealing hardware, reporting minor-embedding behaviour and its practical capacity frontier, a provable non-embeddability bound for the sequence-pair model, a calibrated direct-QPU sweep, and Leap hybrid solves in both penalty and constraint-native form;
(iv) an honest empirical comparison of both formulations under multi-seed classical simulated annealing, simulated quantum annealing, and an exact integer-programming baseline, extended by a six-instance campaign with per-instance calibration, and complemented by an in-depth structural analysis featuring an ablation study to isolate the effect of each penalty under paired significance tests, a sensitivity and calibration-stability analysis, and an automated calibration of the Lagrangian multipliers using a Tree-structured Parzen Estimator search \cite{bergstra2011} over a documented space; and
(v) a practical selection rule that states, as a function of the reduced plate resolution and the piece-count budget, which of the two formulations is preferable. We make no claim of quantum speedup.

\textit{Paper organisation.} Section \ref{sec:sota} reviews related work in classical 2D-CSP solvers and in quantum approaches to packing problems. Section \ref{sec:problem} defines the problem on a fixed plate. Section \ref{sec:qubo} presents both QUBO formulations and their size analysis. Section \ref{sec:feasibility} develops the ground-state feasibility analysis and proves the impossibility result. Section \ref{sec:setup} describes the experimental setup, and Section \ref{sec:results} reports the empirical results from simulated annealing and D-Wave hardware against the integer-programming baseline, alongside the scaling projection. Section \ref{sec:discussion} discusses the selection rule and outlines extensions, and Section \ref{sec:conclusion} concludes.

\section{State of the Art}
\label{sec:sota}
The two-dimensional Cutting Stock Problem sits at the intersection of two largely separate research traditions: classical operations research, which has produced a mature ecosystem of exact and heuristic algorithms over the past six decades, and quantum optimization, which has only recently begun to address rectangular packing problems on programmable quantum annealing hardware. This section reviews both bodies of work, identifies the methodological divide along which the present comparison is organised, and states the specific gap that the paper addresses.

\subsection{Classical Approaches to 2D Cutting Stock}
Following the typology of W\"ascher et al. \cite{waescher2007}, the problem considered in this work falls within the family of two-dimensional Single Large Object Placement Problems with strongly heterogeneous item assortments. This family is known to be NP-hard, and classical solution strategies divide into two complementary camps. Exact methods based on mixed-integer linear programming or branch-and-cut deliver provably optimal solutions but become computationally intractable for instances of industrially relevant size, where the number of candidate piece-placement combinations grows combinatorially. Heuristic methods such as bottom-left, first-fit-decreasing and level-based packing, and meta-heuristics such as simulated annealing, genetic algorithms and tabu search, trade optimality guarantees for scalability and remain the dominant production approach in industry. The structural tension between these two extremes, guaranteed optimality at intractable cost versus tractable cost without quality guarantees, is precisely what motivates the exploration of alternative computing paradigms.

\subsection{Quantum Annealing as an Optimization Paradigm}
Quantum annealing (QA) is a heuristic combinatorial-optimization algorithm rooted in the adiabatic theorem of quantum mechanics. The original proposal of Kadowaki and Nishimori \cite{kadowaki1998} showed that adding a transverse-field driver Hamiltonian to a problem Ising Hamiltonian, and slowly removing the driver, can guide the system toward the ground state encoding the solution of an optimization problem. Modern reviews \cite{hauke2020,yarkoni2022} cover both the theoretical foundations of QA and the spectrum of industrial applications explored on current D-Wave hardware, ranging from scheduling and logistics to finance and machine learning. Crucially, these reviews also highlight the hybrid quantum-classical portfolio of solvers, which accept problems far larger than the bare quantum processing unit.

Problems addressed by QA must be expressed as a Quadratic Unconstrained Binary Optimization (QUBO) instance, or equivalently as an Ising model. The systematic mapping of NP-hard combinatorial problems to this form was catalogued by Lucas \cite{lucas2014}, who provides canonical QUBO formulations for the twenty-one NP-complete problems of Karp and several additional combinatorial instances of practical interest. A central design choice in any such formulation is the encoding of integer-valued auxiliary variables. The standard binary (radix-2) encoding competes with the domain-wall encoding introduced by Chancellor \cite{chancellor2019}, which trades fewer binary variables for higher qubit-connectivity requirements on the underlying hardware.

It is important to note that, while QA is the only currently programmable quantum optimization heuristic at scale, a general quantum speedup over the best classical heuristics has not been empirically established for combinatorial optimization. The seminal study of R{\o}nnow et al. \cite{ronnow2014} introduced a formal methodology for measuring quantum speedup and reported no evidence of such speedup over optimized simulated annealing on random spin-glass benchmarks executed on early D-Wave hardware. Subsequent work has refined both hardware and benchmarks, but the question of when, if ever, current-generation quantum annealers outperform their classical counterparts on industrially relevant problems remains open. The present work makes no claim to settle this question; it contributes a structural characterization that tells a practitioner which of two encodings of one industrial problem class is worth placing on the hardware at all.

A central practical difficulty when deploying QUBO formulations is the calibration of the Lagrangian multipliers that balance constraint penalties against the objective. Manual tuning is the default in the literature \cite{yarkoni2022,okada2024} and limits reproducibility. Cellini et al. \cite{cellini2024} address this through an augmented-Lagrangian framework that derives multipliers analytically. The present work adopts a complementary data-driven approach, using Bayesian optimization with a Tree-structured Parzen Estimator over the multiplier space \cite{bergstra2011,akiba2019}.

\subsection{Quantum Approaches to Packing Problems}
The application of QA to packing problems is a young but active line of research, concentrated mostly on the one-dimensional Bin Packing Problem (1D-BPP). Cellini et al. \cite{cellini2024} proposed QAL-BP, an augmented-Lagrangian QUBO formulation for 1D-BPP in which the penalty multipliers are derived analytically rather than tuned empirically. Garc\'ia-de-Andoin et al. have contributed two complementary studies on the same problem: a hybrid classical-quantum heuristic in which QA samples feasible single-bin configurations that are then assembled by a classical post-processor \cite{garciadeandoin2022a}, and a comparative empirical benchmark against established classical methods \cite{garciadeandoin2022b}. Bozhedarov et al. \cite{bozhedarov2024} apply quantum and quantum-inspired solvers to a real industrial instance of minimum bin packing arising from the storage of spent nuclear fuel. Moving beyond one dimension, Romero et al. \cite{romero2023} solve realistic three-dimensional bin packing instances with the Leap Constrained Quadratic Model hybrid solver, showing that hybrid pipelines already absorb packing workloads of industrial size. However, the geometric reasoning in that line is volumetric rather than positional. Closest to our problem family, Arai and Haraguchi \cite{arai2021} formulate an Ising model for a two-dimensional cutting stock variant whose objective is the minimization of setup cost under a guillotine-style cut structure, evaluated through classical simulation, with execution on a quantum annealer left as future work.

A common feature of this body of work is its restriction to one-dimensional packing, to volumetric accounting, or to structured cut patterns. To the best of our knowledge, prior QA studies on two-dimensional rectangle packing with full positional reasoning \cite{terada2018,okada2024} target the variant with unique items. In contrast, the present work tackles the Cutting Stock variant with free type repetition on a fixed plate and full non-guillotine non-overlap and containment. Consequently, this study appears to be the first to evaluate this variant on real quantum annealing hardware.

\subsection{Encoding the Geometry: Coordinate-Based and Coordinate-Free QUBOs}The sequence-pair representation was introduced by Murata et al. \cite{murata1996} in the context of VLSI module placement. A pair of permutations $(\Gamma_{+},\Gamma_{-})$ on $N$ indices jointly encodes all relative positioning information required to recover a non-overlapping placement of $N$ rectangles: pieces $a$ and $b$ are in the left-of relation when $a$ precedes $b$ in both permutations, and in the below relation when $a$ precedes $b$ in $\Gamma_{+}$ but follows it in $\Gamma_{-}$. The representation is provably complete, in the sense that every non-overlapping placement admits a sequence pair, and conversely, every sequence pair determines a non-overlapping placement. This mapping is recovered through a longest-path computation on the horizontal and vertical constraint DAGs derived from the relations, computable in polynomial time by topological relaxation \cite{cormen2009}. A coordinate-free encoding encodes only the relative order of pieces through two permutations, recovering physical coordinates via a classical longest-path computation. This approach, pioneered by Murata et al. \cite{murata1996}, brought to Ising form by Terada et al. \cite{terada2018}, and refined by Okada et al. \cite{okada2024}, produces a variable count far smaller than the coordinate-based formulation on instances with a large number of candidate positions.

At the opposite extreme sit coordinate-based encodings, which discretize the plate and attach one binary variable to each candidate position of each piece. Non-overlap then becomes a pairwise exclusion between intersecting placements, and the model collapses to a weighted set-packing or maximum-weight independent-set QUBO in the canonical form introduced by Lucas \cite{lucas2014}. The variable count inherits the plate resolution, which is the textbook objection to this encoding, but every constraint of the problem is pairwise by nature, so no information is lost: the ground state is geometrically feasible by construction. The present paper formalises this asymmetry. Okada et al. mitigate the structural loopholes of the coordinate-free encoding with their search-space restriction. We remove that restriction, which conflicts with free piece repetition, and replace it with two structural penalties. We then prove that no replacement of bounded degree can close the gap completely, locating the exact boundary between the two encoding philosophies.

\section{Problem Definition}
\label{sec:problem}

\subsection{The 2D Cutting Stock Problem on a Fixed Plate}
We consider a single-plate 2D Cutting Stock Problem. A rectangular plate of fixed dimensions $W_{P}\times H_{P}$ is given, together with a catalogue of $M$ piece types. Each type $t\in\{1,\dots,M\}$ has integer dimensions $(w_{t},h_{t})$ and may appear any number of times in a valid cutting pattern (free repetition).

A perimeter margin $m\geq 0$ may be reserved around the border of the plate, for example for clamping or cutting tolerance, reducing the usable area to an effective plate of dimensions $W_{P}^{\mathrm{eff}}=W_{P}-2m$ and $H_{P}^{\mathrm{eff}}=H_{P}-2m$. An inter-piece spacing can be folded into the effective piece dimensions in the same way. All pieces must be placed axis-aligned and entirely within the effective area, without overlap. The objective is to select and arrange pieces so as to minimize waste, or equivalently to maximize the total occupied area. Throughout the experimental sections we take $m=0$, so effective and nominal dimensions coincide, and we write $W$, $H$ for the plate dimensions and $A_{\mathrm{plate}}=WH$ for its area.

\subsection{Expanded Catalogue with Rotation}
To support 90-degree rotation, each non-square type generates two entries in an expanded catalogue of $P$ orientations: entry $(t,0)$ with dimensions $(w_{t},\,h_{t})$ and, when $w_{t}\neq h_{t}$, entry $(t,1)$ with dimensions $(h_{t},\,w_{t})$. We index the expanded catalogue by $k\in\{1,\dots,P\}$, write $(\tilde w_{k},\tilde h_{k})$ for the effective dimensions and $a_{k}=\tilde w_{k}\tilde h_{k}$ for the area of entry $k$, and let $a_{\min}=\min_{k}a_{k}$ and $a_{\max}=\max_{k}a_{k}$. The quantity $n_{\max}=\lfloor A_{\mathrm{plate}}/a_{\min}\rfloor$ bounds the number of pieces in any packing.

\subsection{Grid Reduction by Common Divisors}
\begin{lemma}\label{lem:gcd}
Let $g$ divide $W$, $H$, $m$ and every $w_{t}$, $h_{t}$. Then the instance is equivalent to its rescaling by $1/g$: optimal areas correspond under multiplication by $g^{2}$, and an optimal packing of one instance maps to an optimal packing of the other.
\end{lemma}
\noindent\textit{Proof:} In any axis-aligned packing, push every piece maximally to the left and then maximally down; each resulting coordinate is a sum of piece dimensions and hence a multiple of $g$, so an optimal packing exists on the $g$-grid, which is the unit grid of the rescaled instance. \hfill$\blacksquare$

The reduction matters because the variable count of the coordinate-based formulation of Section \ref{sec:qubo} is driven by the number of grid positions, which shrinks by a factor approaching $g^{2}$ on plates large relative to the pieces, while the variable count of the coordinate-free formulation is invariant under it. Lemma \ref{lem:gcd} should therefore always be applied as preprocessing before comparing the two models.

\subsection{Reference Instance}
The reference instance used throughout the paper has a $7\times 5$ plate ($A_{\mathrm{plate}}=35$), margin $m=0$, and the catalogue $\{(4,3),(2,5),(3,2),(2,2),(4,1),(3,1)\}$ of $M=6$ physical types, which expands under rotation to $P=11$ orientations with $a_{\min}=3$, $a_{\max}=12$ and $n_{\max}=11$. The greatest common divisor of all dimensions is $g=1$, so Lemma \ref{lem:gcd} yields no reduction. The exact integer-programming baseline of Section \ref{sec:setup} proves that the optimum is $A^{\star}=35$, a perfect packing, attained for example by two vertical $2\times 5$ columns on $x\in[0,4)$, a rotated $3\times 4$ block on $[4,7)\times[0,4)$ and a $3\times 1$ strip on $[4,7)\times[4,5)$; no three catalogue entries reach area 35, so every optimal packing uses at least $n^{\star}=4$ pieces.

\section{QUBO Formulations}
\label{sec:qubo}

\subsection{Coordinate-Based Set-Packing Formulation}
\label{sec:setpacking}
The coordinate-based model introduces one binary variable per candidate placement. A placement is a pair $p=(k,(u,v))$ of a catalogue entry and an anchor position with $0\leq u\leq W-\tilde w_{k}$ and $0\leq v\leq H-\tilde h_{k}$, so the piece lies entirely inside the plate by construction and containment never needs to be penalised. The placement set $\mathcal{P}$ has cardinality
\begin{equation}\label{eq:vsp}
V_{\mathrm{sp}}=|\mathcal{P}|=\sum_{k=1}^{P}\,(W-\tilde w_{k}+1)(H-\tilde h_{k}+1),
\end{equation}
which is $182$ on the reference instance. Two placements conflict when their rectangles intersect; the conflict set $\mathcal{C}$ contains $7\,005$ unordered pairs on the reference instance. The Hamiltonian is
\begin{equation}\label{eq:hsp}
H_{\mathrm{SP}}=\sum_{p\in\mathcal{P}}\,(\varepsilon-\mu\,a_{p})\,x_{p}\;+\;\lambda\sum_{(p,q)\in\mathcal{C}}x_{p}\,x_{q},
\end{equation}
with $\mu>0$ the area reward, $\lambda>0$ the overlap penalty, and $\varepsilon\geq 0$ an optional cardinality tie-break that prefers packings with fewer, larger pieces among those of equal area. Section \ref{sec:feasibility} proves that $\lambda>\mu\,a_{\max}$ and $0\leq\varepsilon<\mu/n_{\max}$ make the model exact: every ground state is a geometrically feasible maximum-area packing. This is a weighted set-packing QUBO in the sense of Lucas \cite{lucas2014}. Its decoder is trivial, since active variables carry their own coordinates.

Two structural properties of the conflict graph matter for hardware. First, conflicts are local: a placement intersects only placements anchored within one piece diameter of it, so the degree of a placement is bounded by a catalogue-dependent constant independent of the plate size. On the reference instance the maximum degree is $140$, and on plates from $15\times 11$ upward it saturates at $289$ for this catalogue (Section \ref{sec:results}). Second, the graph contains no global structure: its cliques are the sets of placements covering a single cell, again of catalogue-bounded size. The model therefore stays sparse as the plate grows, with $7\,187$ nonzero terms ($182$ linear plus $7\,005$ quadratic) on the reference instance.

\subsection{Coordinate-Free Sequence-Pair Formulation}
\label{sec:seqpair}
The coordinate-free model operates on an $N\times N$ oblique grid, where $N$ is a user-defined budget on the number of simultaneously placed pieces. Each grid cell is a pair $(i,j)$ with $i,j\in\{0,\dots,N-1\}$: the row index $i$ encodes the position in $\Gamma_{+}$ and the column index $j$ the position in $\Gamma_{-}$, so a valid configuration activates exactly one cell per row and per column, realising a permutation $\sigma$. For two occupied cells $a=(i_{a},j_{a})$ and $b=(i_{b},j_{b})$ with $i_{a}<i_{b}$, the relation follows from the columns alone: $j_{a}<j_{b}$ places $a$ to the left of $b$, and $j_{a}>j_{b}$ places $a$ below it. Physical coordinates are not encoded. They are recovered classically by a longest-path relaxation on the horizontal and vertical constraint DAGs induced by these relations, in $O(N^{2})$ time after a topological sort \cite{cormen2009}, which produces the tightest packing consistent with the sequence pair. Completeness of the representation \cite{murata1996} guarantees that decoded layouts never overlap. Only containment can fail, the requirement that $X_{s}+\tilde w_{k_{s}}\leq W$ and $Y_{s}+\tilde h_{k_{s}}\leq H$ for every occupied slot $s$. As Section \ref{sec:feasibility} shows, this failure mode is unavoidable.

\textbf{Decision variables.} For every cell $(i,j)$ and every catalogue entry $k\in\{0,1,\dots,P\}$, where $k=0$ is an empty marker of zero dimensions, a placement bit $x_{ijk}$ states that the cell hosts entry $k$. Following \cite{okada2024}, envelope-chain bits $y_{ijk}$ and $z_{ijk}$, defined for real entries $k\geq 1$ only, state that the piece at $(i,j)$ belongs to the horizontal (respectively vertical) boundary chain of the layout. Slack bits $s^{W}_{r}$ ($r<R_{W}$), $s^{H}_{r}$ ($r<R_{H}$) and $s^{A}_{r}$ ($r<R_{A}$) absorb, in radix-2, the deficits of the two boundary equalities and of the area budget. The variable count is
\begin{equation}\label{eq:vseq}
V_{\mathrm{seq}}=N^{2}(P+1)+2N^{2}P+R_{W}+R_{H}+R_{A},
\end{equation}
which for the reference instance with $N=7$ and $P=11$ evaluates to $588+1\,078+14=1\,680$.

\textbf{Normalization.} All dimensional quantities entering the boundary penalties are normalised by the catalogue perimeter sum
\begin{equation}\label{eq:norm}
\begin{split}
S &= \sum_{k=1}^{P}\bigl(\tilde w_{k}+\tilde h_{k}\bigr), \qquad \hat w_{k}=\tilde w_{k}/S, \\
  &\quad \hat h_{k}=\tilde h_{k}/S, \quad \hat W_{P}=W/S, \quad \hat H_{P}=H/S,
\end{split}
\end{equation}
so that a single set of penalty multipliers remains approximately transferable across instances of different physical scale. On the reference instance, $S=60$. Area quantities ($a_{k}$, $A_{\mathrm{plate}}$) are kept in plate units so that the area slack quantum is exactly one unit of area.

\textbf{Permutation penalty.} Every row and every column must contain exactly one active cell, a real piece or the empty marker:
\begin{equation}\label{eq:hperm}
H_{\mathrm{perm}}=\sum_{i=0}^{N-1}\Bigl(\sum_{j,k}x_{ijk}-1\Bigr)^{2}+\sum_{j=0}^{N-1}\Bigl(\sum_{i,k}x_{ijk}-1\Bigr)^{2}.
\end{equation}

\textbf{Chain consistency.} Envelope membership must imply placement, as a one-directional implication ($y_{ijk}\leq x_{ijk}$ and analogously for $z$):
\begin{equation}\label{eq:hcons}
H_{\mathrm{cons}}=\sum_{i,j}\sum_{k\geq 1}\bigl[\,y_{ijk}\,(1-x_{ijk})+z_{ijk}\,(1-x_{ijk})\,\bigr].
\end{equation}
A placed piece is not forced into a chain. The area and pairwise penalties below only partially address the structural laxity this leaves.

\textbf{Chain-path validity.} Members of the horizontal envelope must be pairwise in the left-of relation, and members of the vertical envelope pairwise in the below relation. Writing $\mathcal{B}$ (respectively $\mathcal{L}$) for the set of unordered cell pairs in the below (left-of) relation, the conflicting memberships are penalised:
\begin{equation}\label{eq:hpath}
H_{\mathrm{path}}=\sum_{(a,b)\in\mathcal{B}}\sum_{k_{a},k_{b}\geq 1}y_{ak_{a}}\,y_{bk_{b}}+\sum_{(a,b)\in\mathcal{L}}\sum_{k_{a},k_{b}\geq 1}z_{ak_{a}}\,z_{bk_{b}},
\end{equation}
with the shorthand $y_{ak}=y_{i_{a}j_{a}k}$. Pairs sharing a row or a column cannot be simultaneously occupied once $H_{\mathrm{perm}}$ is satisfied, so only the opposite relation needs penalising.

\textbf{Boundary equalities.} The horizontal envelope spans the width of the layout and the vertical envelope its height. Their normalised extents must match the plate up to slack. With slack quantum $\Delta_{W}=\hat w_{\min}/2$, half the smallest normalised width, and $R_{W}=\min\{6,\lceil\log_{2}(\hat W_{P}/\Delta_{W}+1)\rceil\}$ bits (and the analogous $\Delta_{H}$, $R_{H}$),
\begin{equation}\label{eq:hbnd}
H^{W}_{\mathrm{bnd}}=\Bigl(\sum_{i,j,k\geq 1}y_{ijk}\,\hat w_{k}+\Delta_{W}\sum_{r=0}^{R_{W}-1}2^{r}s^{W}_{r}-\hat W_{P}\Bigr)^{2},
\end{equation}
and symmetrically for $H^{H}_{\mathrm{bnd}}$ with $z$, $\hat h_{k}$, $\hat H_{P}$. With integer plate and piece dimensions, $\hat W_{P}$ is an integer multiple of $\Delta_{W}$ (here $14\Delta_{W}$) within the 4-bit range, so the residual can vanish exactly. Chain over-length cannot be absorbed and is therefore penalised quadratically.

\textbf{Total-area constraint.} A configuration may activate placement bits outside any envelope chain, gaining area reward while every penalty above stays at zero. The area budget closes this orphan loophole:
\begin{equation}\label{eq:harea}
H_{\mathrm{area}}=\Bigl(\sum_{i,j,k\geq 1}x_{ijk}\,a_{k}+\sum_{r=0}^{R_{A}-1}2^{r}s^{A}_{r}-A_{\mathrm{plate}}\Bigr)^{2},
\end{equation}
with $R_{A}=\lceil\log_{2}(A_{\mathrm{plate}}+1)\rceil$ ($=6$ here), written in plate units so that any decorated area up to $A_{\mathrm{plate}}$ is matched exactly and any excess is penalised quadratically.

\textbf{Pairwise geometric infeasibility.} Two pieces in the below relation whose heights sum beyond the plate, or in the left-of relation with widths beyond the plate, necessarily overflow after longest-path recovery, yet incur no penalty so far. They are suppressed directly:
\begin{equation}\label{eq:hpair}
\begin{split}
H_{\mathrm{pair}}={}&\sum_{(a,b)\in\mathcal{B}}\;\sum_{\substack{k_{a},k_{b}\geq 1\\ \tilde h_{k_{a}}+\tilde h_{k_{b}}>H}}x_{ak_{a}}\,x_{bk_{b}}\\
&+\sum_{(a,b)\in\mathcal{L}}\;\sum_{\substack{k_{a},k_{b}\geq 1\\ \tilde w_{k_{a}}+\tilde w_{k_{b}}>W}}x_{ak_{a}}\,x_{bk_{b}}.
\end{split}
\end{equation}
This is by construction a pairwise approximation: chains of three or more pieces where no single pair overflows but the cumulative sum does are not captured. Section \ref{sec:feasibility} proves that no bounded-degree completion exists. 

Neither $H_{\mathrm{area}}$ nor $H_{\mathrm{pair}}$ belongs to the core representation: both are structural reinforcements layered on top of it, included because removing them measurably degrades the feasibility of the decoded samples, an effect isolated in the ablation study of Section \ref{sec:results}.

\textbf{Objective and total Hamiltonian.} The objective rewards decorated area, $H_{\mathrm{obj}}=\sum_{i,j,k\geq 1}x_{ijk}\,a_{k}$, and the total Hamiltonian is
\begin{equation}\label{eq:htotal}
\begin{split}
H_{\mathrm{total}} &= \alpha H_{\mathrm{perm}}+\beta H_{\mathrm{cons}}+\gamma H_{\mathrm{path}}+\delta\bigl(H^{W}_{\mathrm{bnd}}+H^{H}_{\mathrm{bnd}}\bigr) \\
&\quad +\delta_{a}H_{\mathrm{area}}+\lambda H_{\mathrm{pair}}-\mu H_{\mathrm{obj}},
\end{split}
\end{equation}
where every penalty is non-negative and vanishes exactly when its constraint is satisfied. The seven multipliers are calibrated automatically by the TPE protocol of Section \ref{sec:setup}. The selected values are $\alpha=39.46$, $\beta=6.98$, $\gamma=16.83$, $\delta=4.13$, $\delta_{a}=0.0582$, $\lambda=25.84$, $\mu=0.0287$.

\begin{table}[!t]
\centering\footnotesize
\caption{Quadratic couplers contributed by each penalty block of the sequence-pair Hamiltonian on the reference instance ($N=7$, $P=11$), before deduplication. Blocks share couplers: every $H_{\mathrm{path}}$ coupler also appears in $H_{\mathrm{bnd}}$, every $H_{\mathrm{pair}}$ coupler in $H_{\mathrm{area}}$, and $H_{\mathrm{perm}}$ and $H_{\mathrm{area}}$ share $38\,269$ pairs. The union holds $450\,925$ distinct couplers over $1\,680$ variables (density $0.32$, maximum degree $559$). The two global equality penalties expand to complete graphs.}
\label{tab:blocks}
\begin{tabular}{lrl}
\toprule
Block & Couplers & Structure \\
\midrule
$H_{\mathrm{perm}}$ & 45\,570 & row and column one-hot cliques \\
$H_{\mathrm{cons}}$ & 1\,078 & $y$--$x$ and $z$--$x$ implication pairs \\
$H_{\mathrm{path}}$ & 106\,722 & relation-conflicting chain pairs \\
$H_{\mathrm{bnd}}^{W}+H_{\mathrm{bnd}}^{H}$ & 294\,306 & two equality cliques $K_{543}$ \\
$H_{\mathrm{area}}$ & 148\,240 & one equality clique $K_{545}$ \\
$H_{\mathrm{pair}}$ & 32\,193 & per-pair infeasible placements \\
\midrule
Union (distinct) & 450\,925 & $+\,1\,680$ linear terms \\
\bottomrule
\end{tabular}
\end{table}

\subsection{Size Comparison on the Reference Instance}
\label{sec:sizes}
Table \ref{tab:blocks} breaks the sequence-pair coupler count down by block. The decisive observation is structural: the two boundary equalities and the area budget are global constraints, so their squared penalties couple every pair of participating variables, expanding to complete graphs $K_{543}$ (each boundary: $539$ chain bits plus $4$ slack bits) and $K_{545}$ ($539$ placement bits plus $6$ slack bits). Together with $H_{\mathrm{path}}$, they account for over 95 percent of the $450\,925$ distinct couplers, against $32\,193$ for the pairwise containment penalty $H_{\mathrm{pair}}$, which adds no variables (7.1 percent of the couplers). On the reference instance, the sequence-pair model therefore requires 9.2 times the variables and 64 times the couplers of the set-packing model, resulting in a maximum degree of 559 versus 140. Section \ref{sec:results} turns these counts into closed-form scaling laws and into a provable minor-embedding separation.

\section{Ground-State Feasibility Analysis}
\label{sec:feasibility}

\subsection{Exactness of the Coordinate-Based Model}
\begin{proposition}\label{prop:exact}
Let $H_{\mathrm{SP}}$ be the Hamiltonian \eqref{eq:hsp} with $\lambda>\mu\,a_{\max}$ and $0\leq\varepsilon<\mu/n_{\max}$. Then every ground state of $H_{\mathrm{SP}}$ is a geometrically feasible packing of maximum area $A^{\star}$, and the ground states are exactly the maximum-area packings of minimum cardinality.
\end{proposition}
\noindent\textit{Proof:} Suppose a configuration activates two conflicting placements and let $p$ be one of them, involved in $c_{p}\geq 1$ active conflicts. Deactivating $p$ changes the energy by $(\mu a_{p}-\varepsilon)-\lambda c_{p}\leq\mu a_{\max}-\lambda<0$, so no configuration with a conflict is a ground state. Conflict-free configurations are packings, with energy $\varepsilon n-\mu A$ for cardinality $n$ and area $A$. Let $n^{\star}$ be the minimum cardinality among packings of area $A^{\star}$. For any packing with $A<A^{\star}$, integrality of areas gives $(\varepsilon n^{\star}-\mu A^{\star})-(\varepsilon n-\mu A)\leq \varepsilon\,n_{\max}-\mu<0$; for $A=A^{\star}$ the difference is $\varepsilon(n^{\star}-n)\leq 0$, with equality exactly when $n=n^{\star}$. \hfill$\blacksquare$

\begin{remark}\label{rem:eps}
The tuned values of Section \ref{sec:setup} are $\mu=1$, $\lambda=36.18>\mu a_{\max}=12$ and $\varepsilon=0.360$, which exceeds the universal bound $\mu/n_{\max}\approx 0.091$. Exactness nevertheless holds on the reference instance for any $\varepsilon<\mu$: no two catalogue areas sum beyond $24$, so every packing of area $34$ uses at least three pieces while the optimum uses $n^{\star}=4$, and a case check over the attainable (area, cardinality) pairs gives $\varepsilon(n^{\star}-n)<\mu(A^{\star}-A)$ whenever $A<A^{\star}$. The bound $\varepsilon<\mu/n_{\max}$ remains the safe instance-independent choice, and the SA, SQA and hybrid runs of Section \ref{sec:results} all confirm a feasible optimal ground state (the hybrid energy $-33.559=-\mu A^{\star}+\varepsilon n^{\star}$ matches the prediction to machine precision).
\end{remark}

\subsection{An Impossibility Result for Coordinate-Free Encodings}
\begin{definition}\label{def:faithful}
Consider a coordinate-free formulation whose binary configurations assign to each of $N$ slots a catalogue entry or the empty marker together with pairwise sequence-pair relations, decoded by longest paths, and write $A(\mathbf{x})$ for the decorated area, the total area of the placed entries. An energy $f=\mathcal{P}-\mu A$ with $\mu>0$ is an \emph{exact penalty encoding of degree $d$} if $\mathcal{P}$ is a pseudo-Boolean polynomial of degree at most $d$ in the configuration bits and $\mathcal{P}(\mathbf{x})=0$ for every configuration $\mathbf{x}$ whose decoded layout satisfies containment.
\end{definition}
The definition captures the standard Lagrangian design, in which each penalty term is non-negative and vanishes exactly on satisfied constraints, so that the energy ranks feasible layouts purely by area; the Hamiltonian \eqref{eq:htotal} is of this form, since $H_{\mathrm{pair}}$ only fires on pairs that are infeasible on their own.

\begin{theorem}\label{thm:impossible}
For every $d\geq 1$ there is an instance of the fixed-plate 2D-CSP on which every exact penalty encoding of degree at most $d$ has all of its minimum-energy configurations geometrically infeasible.
\end{theorem}
\noindent\textit{Proof:} Fix $d$ and take the instance with plate $1\times H$, a single catalogue type of dimensions $1\times h$ with $d\,h\leq H<(d+1)\,h$, and slot budget $N\geq d+1$. Fix a configuration template in which slots $1,\dots,d+1$ form a chain under the below relation and every remaining slot carries the empty marker, and for $S\subseteq\{1,\dots,d+1\}$ let $\mathbf{x}_{S}$ be the configuration placing the piece exactly in the slots of $S$. Toggling slot $s$ between the piece and the empty marker is affine in a single bit $t_{s}$ (its one-hot block reads $x_{s,\mathrm{piece}}=t_{s}$ and $x_{s,0}=1-t_{s}$, all other bits constant), so $g(t):=\mathcal{P}(\mathbf{x}_{S(t)})$ is a multilinear polynomial of degree at most $d$ on $\{0,1\}^{d+1}$. If some coordinate of $t$ vanishes, at most $d$ pieces are placed. Their longest-path ordinates stack to at most $d\,h\leq H$, the layout is feasible, and exactness gives $g(t)=0$. A multilinear function on $\{0,1\}^{d+1}$ vanishing at every vertex except possibly the all-ones vertex is a scalar multiple of $t_{1}t_{2}\cdots t_{d+1}$, of degree $d+1$. Since $\deg g\leq d$, $g$ vanishes identically, and in particular $\mathcal{P}(\mathbf{x}_{\{1,\dots,d+1\}})=0$. That configuration stacks $d+1$ pieces of total height $(d+1)h>H$, hence is infeasible, with energy $-\mu(d+1)h$. Because the plate has unit width, any two placed pieces of a feasible layout must be in the below relation, so a feasible configuration carries at most $\lfloor H/h\rfloor=d$ pieces and has energy at least $-\mu\,d\,h>-\mu(d+1)h$. The minimum of $f$ therefore lies strictly below every feasible configuration, and no minimizer is feasible. \hfill$\blacksquare$

\begin{example}\label{ex:d2}
For $d=2$, the quadratic case of QUBO hardware, take the plate $1\times 5$ and three pieces $1\times 2$: every pair stacks to height $4\leq 5$, the triple to $6>5$. Any quadratic penalty that vanishes on feasible layouts assigns the overflowing triple zero penalty as well, and the triple, carrying the largest area reward, becomes the ground state.
\end{example}

\begin{remark}\label{rem:soft}
Dropping exactness can rescue individual instances. On Example \ref{ex:d2}, a uniform inexact penalty $p$ on every below-related pair of placed pieces makes the feasible two-piece stack the unique ground state precisely when $p\in(\mu a/2,\,\mu a)$ with $a$ the piece area. Two costs follow. First, the admissible window depends on the piece areas and chain lengths of the instance, so the weights must be retuned per instance, which is what our TPE protocol does in practice. Second, faithfulness is lost: the optimal feasible layout is itself penalised, its energy no longer equals $-\mu A$, and the energy ranking of feasible layouts is distorted. This is the mechanism behind the weak energy-area correlation measured for the sequence-pair model in Section \ref{sec:results} ($-0.49$, against $-0.99$ for the exact set-packing model). The implemented $H_{\mathrm{pair}}$ deliberately avoids this trade by penalising only pairs that are infeasible on their own, so the formulation \eqref{eq:htotal} remains exact in the sense of Definition \ref{def:faithful} and Theorem \ref{thm:impossible} applies to it in full.
\end{remark}

\subsection{An Explicit Infeasible Ground State of the Implemented Hamiltonian}
\begin{proposition}\label{prop:certificate}
Consider the Hamiltonian \eqref{eq:htotal} on the reference instance with any multipliers satisfying $\mu>0$ and $\delta_{a}>\mu$ (the tuned values give $\delta_{a}=0.0582>\mu=0.0287$). Its minimum energy equals $-\mu A_{\mathrm{plate}}$ up to the constants of the squared penalties, and the set of minimizers contains geometrically infeasible configurations. One of them places, under the identity permutation, the entries $4\times 3$, $3\times 1$, $2\times 5$ and $2\times 5$ in the first four slots, leaves the remaining slots empty, sets all envelope-chain bits to zero, fills the boundary slacks to $\hat W_{P}$ and $\hat H_{P}$, and sets the area slack to zero.
\end{proposition}
\noindent\textit{Proof:} Lower bound: every penalty block is non-negative. If the decorated area is $A_{\mathrm{plate}}+o$ with $o\geq 1$, then $H_{\mathrm{area}}\geq\delta_{a}o^{2}$ irrespective of the slack, so $f\geq\delta_{a}o^{2}-\mu(A_{\mathrm{plate}}+o)\geq-\mu A_{\mathrm{plate}}+(\delta_{a}-\mu)o>-\mu A_{\mathrm{plate}}$. If the decorated area is at most $A_{\mathrm{plate}}$, then $f\geq-\mu A_{\mathrm{plate}}$ directly. Attainment: in the displayed configuration, $H_{\mathrm{perm}}$ vanishes (one entry per row and column), $H_{\mathrm{cons}}$ and $H_{\mathrm{path}}$ vanish (no chain bits set), both boundary residues vanish because $\hat W_{P}$ and $\hat H_{P}$ are integer multiples ($14$ and $10$) of the slack quanta $\Delta_{W}=\Delta_{H}=1/120$ within the 4-bit ranges, $H_{\mathrm{area}}$ vanishes because the decorated area is $12+3+10+10=35=A_{\mathrm{plate}}$ with zero slack, and $H_{\mathrm{pair}}$ vanishes because the pairwise width sums along the left-of chain are $7,6,6,5,5,4$, none exceeding $W=7$ (empty slots contribute zero width). The energy is therefore exactly $-35\mu$. The longest-path abscissae of the four pieces are $0$, $4$, $7$ and $9$, so the third and fourth pieces end at $9$ and $11$, beyond $W=7$: the configuration is infeasible. Finally, the geometrically optimal packing of the same multiset (Section \ref{sec:problem}) with the same auxiliary choices also attains $-35\mu$, so the ground manifold is degenerate between feasible and infeasible layouts. \hfill$\blacksquare$

\begin{figure}[!t]
\centering
\includegraphics[width=\columnwidth]{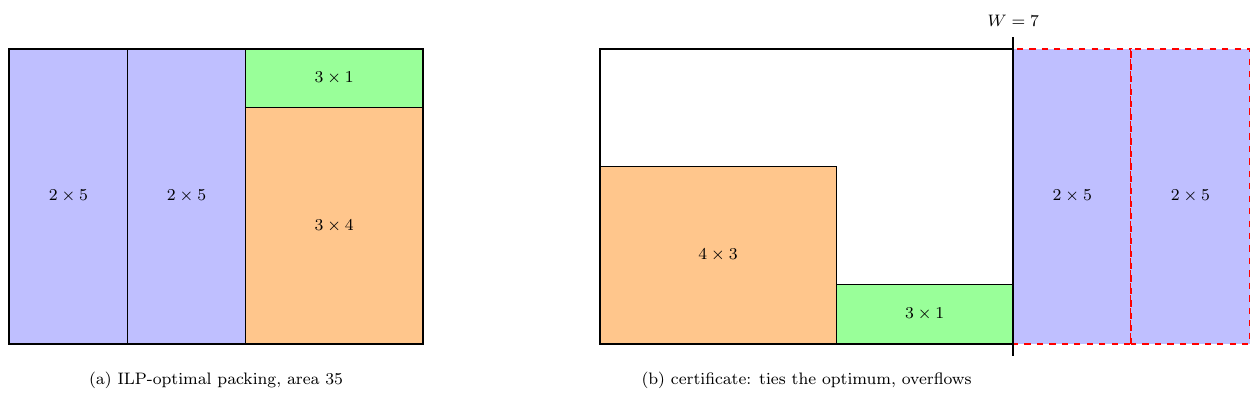}
\caption{The degenerate ground manifold of Proposition \ref{prop:certificate} on the reference instance. Left: an ILP-optimal packing, area 35 with four pieces. Right: the certificate, the same multiset arranged as a single left-of chain with empty envelope chains and saturated slacks; every penalty term vanishes and the energy ties the optimum exactly, but the longest-path abscissae $0, 4, 7, 9$ push the last two pieces (outlined) past the plate edge at $W=7$.}
\label{fig:certificate}
\end{figure}

Proposition \ref{prop:certificate} is Theorem \ref{thm:impossible} made concrete on the tuned model: the certificate is an overflowing left-of chain whose every pair fits. It predicts exactly the empirical behaviour of Section \ref{sec:results}: in every simulated-annealing seed the lowest-energy sample attains the degenerate ground energy and is geometrically infeasible (the infeasible side of the manifold and its low-energy neighbourhood is combinatorially far richer than the feasible side), and the hybrid solver returns the phenomenon in its purest form, a single configuration tying the certificate energy to within $3\cdot 10^{-13}$ whose decoded layout overflows the plate. Fig. \ref{fig:certificate} draws the certificate next to the feasible optimum it ties. The correct reading of the coordinate-free model is therefore sample-and-filter: by completeness of the representation the decoder never produces overlaps, only containment can fail, so filtering decoded samples costs $O(N^{2})$ each and is adequate. The model is a legitimate heuristic sampler, not an exact encoding.

\subsection{The Cost of Exact Repair}
To forbid every overflowing chain of $\ell$ pieces directly, one needs penalty monomials of degree $\ell$ supported on the cell tuples in chain relation and their type assignments, of order $O(N^{2\ell}P^{\ell})$ terms. On the reference instance $\ell=3$ already gives about $N^{6}P^{3}\approx 1.6\cdot 10^{8}$ monomials, and mapping each degree-$\ell$ term to quadratic hardware requires auxiliary variables through standard quadratization \cite{boros2002}. Exact repair is therefore out of reach for near-term annealers, which is Theorem \ref{thm:impossible} read constructively: the missing information is the longest-path accumulation, and it cannot be compressed into bounded-degree interactions.

\section{Experimental Setup}
\label{sec:setup}

\subsection{Instance and Software Stack}
Except where noted, all experiments use the reference instance of Section \ref{sec:problem}, and every number reported in Section \ref{sec:results} was produced in a single consolidated session on one workstation, so that model sizes, samples and wall-clock times are mutually consistent. The software stack is Python 3.12 with dimod 0.12.21 and neal 0.6.0 for classical simulated annealing, OpenJij for simulated quantum annealing, minorminer 0.2.21 and dwave-system 1.34.0 for hardware access (a DWaveSampler attached to the Advantage\_system4 QPU, 5\,627 active qubits, plus the LeapHybridBQMSampler and LeapHybridCQMSampler), PuLP 3.3.2 with the CBC solver \cite{cbc2005} for the exact baseline, and Optuna 4.8.0 \cite{akiba2019} for multiplier calibration.

\subsection{Multiplier Calibration}
Both formulations are calibrated automatically with a Tree-structured Parzen Estimator \cite{bergstra2011}. For the set-packing model, $\mu=1$ is fixed and the TPE explores $\lambda\in[1.5\,a_{\max},\,8\,a_{\max}]$ log-uniformly and $\varepsilon\in[0,0.5]$ over 40 trials, each scored on a 1\,500-read, 600-sweep SA run by the geometric feasibility rate plus the normalised area at minimum energy; the selected values are $\lambda=36.18$ and $\varepsilon=0.360$ (see Remark \ref{rem:eps}). For the sequence-pair model, the seven multipliers of \eqref{eq:htotal} are searched jointly over log-uniform ranges $\alpha\in[5,80]$, $\beta\in[1,30]$, $\gamma\in[2,50]$, $\delta\in[0.5,15]$, $\delta_{a}\in[0.005,0.5]$, $\lambda\in[3,60]$ and $\mu\in[0.005,0.15]$, each trial scored on a 3\,000-read, 1\,000-sweep SA run by the normalised area at minimum energy plus a saturating bonus in the count of geometrically feasible samples. The configuration used throughout (quoted in Section \ref{sec:seqpair}) was selected by an earlier study under this protocol, and a fresh 60-trial study over the documented space, taking 4.9 h, reproduces its quality, as Section \ref{sec:results} reports.

\subsection{Protocols and Metrics}
The final SA evaluation uses eleven seeds: 42, 123, 7, 99, 314, 271, 1000, 31415, 8192, 1729 and 12345. The set-packing model is sampled at 4\,000 reads $\times$ 1\,500 sweeps, and the larger sequence-pair model at 8\,000 $\times$ 2\,000. The ablation runs four regimes of \eqref{eq:htotal}, namely the baseline without the two structural penalties, $+H_{\mathrm{area}}$ only, $+H_{\mathrm{pair}}$ only, and the full model. It uses fifteen seeds, the eleven above plus 11, 22, 33 and 44, at 5\,000 reads per seed and with every other multiplier fixed at its tuned value, and compares each regime against the baseline by a paired two-sided Wilcoxon signed-rank test. The sensitivity analysis perturbs each multiplier by $\pm 25$ percent in isolation at seed 42, with 6\,000 reads $\times$ 1\,500 sweeps.

Simulated quantum annealing uses the OpenJij SQASampler. The set-packing model runs at 1\,000 reads $\times$ 1\,500 sweeps. The sequence-pair model gets a reinforced schedule of 1\,000 reads $\times$ 8\,000 sweeps, with $\beta=50$, $\gamma=2$ and eight Trotter slices.

The hardware runs proceed in three stages: minor embedding, direct QPU sampling, and the hybrid solvers. Minor embedding uses minorminer against the working graph of the QPU, with a 120 s budget for the head-to-head embeddings and 180 s and three tries per plate for the capacity sweep. Direct QPU sampling then draws 1\,000 reads per setting on a fixed embedding, sweeping chain strengths $\{20, 40, 60, 80\}$ plus the uniform-torque-compensation value, at annealing times of 20, 100 and 200 $\mu$s. Finally, the hybrid BQM and CQM samplers run at their default time limits. The CQM reformulations carry the constraints of Section \ref{sec:qubo} natively, with no penalty weights and no slack bits: one placement per cell for set-packing, and the one-hot, implication, chain-path, pairwise, boundary and area constraints for sequence-pair, 1\,099 in total.

The exact baseline maximises $\sum_{p}a_{p}x_{p}$ over the 182 placement binaries, subject to at most one placement covering each of the 35 cells. It is solved by CBC with a proof of optimality, and timed as the median of five runs.

The multi-instance campaign of Section \ref{sec:results} runs six instances. Each is preprocessed by Lemma \ref{lem:gcd}, solved exactly by the same ILP, and then evaluated under both formulations with eleven seeds per model. The set-packing model uses the analytic multipliers of Proposition \ref{prop:exact} with no tuning at all ($\mu=1$, $\lambda=2a_{\max}$, $\varepsilon=0.9/n_{\max}$) at 3\,000 reads $\times$ 1\,200 sweeps. The sequence-pair model is recalibrated per instance, as Remark \ref{rem:soft} requires: 25 TPE trials at 2\,000 $\times$ 800, followed by the multi-seed run at 5\,000 $\times$ 1\,500, with the slot budget set to two more than the piece count of the ILP solution (minimum six).

A sample is structurally (QUBO-) feasible when the one-hot, implication and boundary-capacity constraints of Section \ref{sec:seqpair} hold. The set-packing model has no structural constraints, so every assignment is directly decodable. Decoded layouts come from the longest-path post-processor, and a sample is geometrically feasible when every decoded piece satisfies containment. Per seed we report the geometric feasibility rate, the best feasible area, the decoded area of the lowest-energy feasible sample together with the within-pool gap between the two, whether the overall lowest-energy sample is geometrically feasible, and the Pearson correlation between energy and decoded area across feasible samples. The ablation additionally tracks the overshoot rate, the fraction of structurally valid samples whose decorated area exceeds the plate.

\section{Results}
\label{sec:results}
 
\subsection{Simulated Annealing Head-to-Head}
Table \ref{tab:comparison} summarises the multi-seed comparison, and Fig. \ref{fig:scatter} shows the pooled energy-area scatters behind its correlation row. The set-packing model attains the proven optimum $A^{\star}=35$ in all eleven seeds, every returned sample is conflict-free, the minimum-energy sample is feasible in every seed, consistent with Proposition \ref{prop:exact}, the pooled energy-area correlation is $-0.993$, and the pooled runs contain 61 distinct optimal compositions. The sequence-pair model produces structurally valid samples at 100 percent, but only $10.7\pm 0.3$ percent of them decode inside the plate. The average best feasible area is $33.91$, with four of eleven seeds reaching 35 (seeds 42, 123, 8192 and 12345, contributing five distinct optimal compositions), and the per-seed correlation ranges between $-0.43$ and $-0.56$. In every seed the lowest-energy sample attains the degenerate ground energy, matching the certificate value of Proposition \ref{prop:certificate} to within $8\cdot 10^{-6}$ (single-precision accumulation in the sampler), and decodes outside the plate, exactly as the proposition predicts. The within-pool gap is zero in every seed for both models: among the geometrically feasible samples the lowest-energy one already carries the best feasible area, so the failure of the coordinate-free model is not the ranking among feasible layouts but the mass of infeasible states tied with and below the feasible optimum.
 
\begin{table}[!t]
\centering\footnotesize
\caption{Head-to-head comparison of the two QUBO formulations on the fixed-plate reference instance ($7\times 5$, area 35). Simulated-annealing metrics are average $\pm$ standard deviation over the eleven seeds of Section \ref{sec:setup}. The set-packing model has no structural constraints, so every assignment is decodable. The headline row is the geometric feasibility of the minimum-energy sample. All values were produced in the single consolidated session of Section \ref{sec:setup}. Best feasible area is the per-seed average, bounded above by $A^{\star}=35$.}
\label{tab:comparison}
\begin{tabular}{lrr}
\toprule
 & Sequence-pair & Set-packing \\
\midrule
Binary variables & 1\,680 & 182 \\
Quadratic couplers & 450\,925 & 7\,005 \\
Maximum degree & 559 & 140 \\
Interaction-graph density & 0.32 & 0.43 \\
Min.-energy sample feasible & 0/11 seeds & 11/11 seeds \\
Structurally valid samples (\%) & $100.0\pm 0.0$ & -- \\
Geometric feasibility (\%) & $10.7\pm 0.3$ & $100.0\pm 0.0$ \\
Best feasible area & $33.91$ & $35.00$ \\
Seeds reaching $A^{\star}$ & 4/11 & 11/11 \\
corr(energy, area) & $-0.488$ & $-0.993$ \\
\midrule
ILP optimum $A^{\star}$ (CBC) & 35 & 35 \\
\bottomrule
\end{tabular}
\end{table}
 
\begin{figure}[!t]
\centering
\includegraphics[width=0.49\columnwidth]{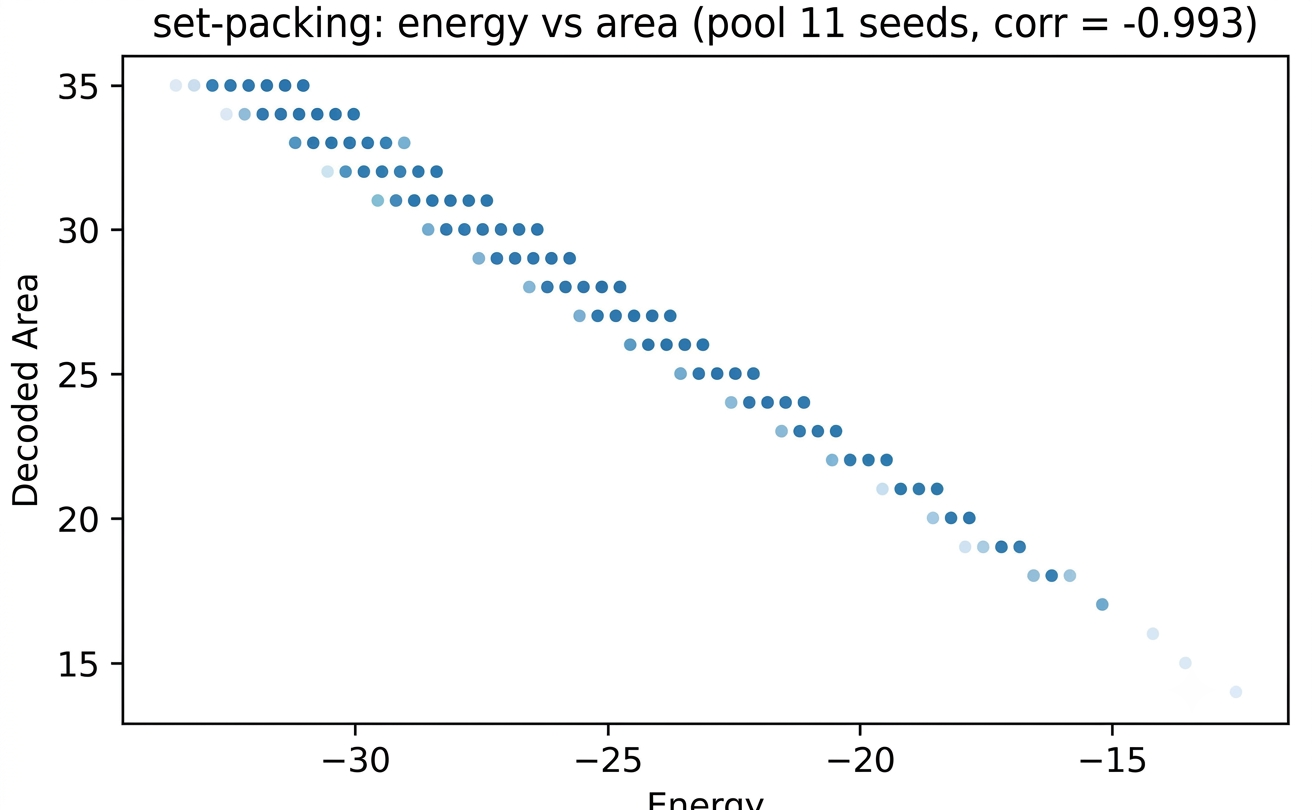}\hfill
\includegraphics[width=0.49\columnwidth]{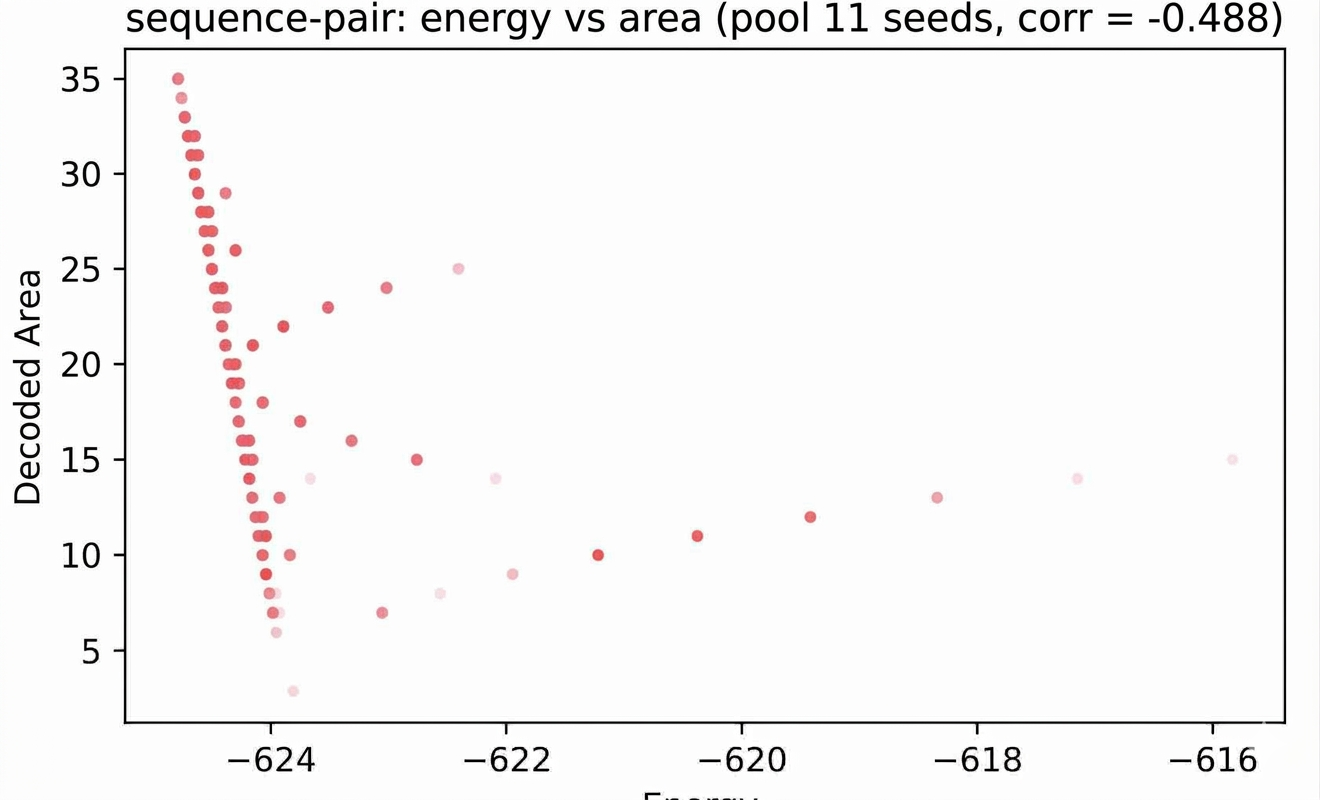}
\caption{Energy against decoded area over the geometrically feasible samples of the pooled eleven-seed SA runs. Left, set-packing: the diagonal staircase is $E=-\mu A+\varepsilon n$, and the small horizontal spread inside each area level is the cardinality tie-break $\varepsilon$ separating compositions of equal area (Pearson correlation $-0.993$). Right, sequence-pair: the zero-penalty manifold compresses the whole feasible area range into a wall of width $\mu(A^{\star}-a_{\min})\approx 0.9$ energy units beside the penalty constants, while residual penalties from mismatched slack registers displace samples rightward by up to nine units regardless of their area. Energy is dominated by constraint bookkeeping rather than by the objective (correlation $-0.488$), the price of a calibration that must keep $\mu$ small against the penalty scales.}
\label{fig:scatter}
\end{figure}
 
\subsection{Ablation of the Structural Penalties}
Table \ref{tab:ablation} isolates the contribution of $H_{\mathrm{area}}$ and $H_{\mathrm{pair}}$ over fifteen seeds. Relative to the baseline, geometric feasibility and overshoot improve under every regime (paired Wilcoxon, $p\leq 6.5\cdot 10^{-4}$ throughout, reaching the smallest attainable two-sided value $6.1\cdot 10^{-5}$ at $n=15$ for $+H_{\mathrm{area}}$). The average best feasible area improves significantly for $+H_{\mathrm{pair}}$ ($p=9.6\cdot 10^{-4}$) and for the full model ($p=1.4\cdot 10^{-3}$) but not for $+H_{\mathrm{area}}$ alone ($p=0.18$). The pairwise penalty is the main driver of solution quality, the area budget the main suppressor of over-decoration, and the combination maximises feasibility (a twenty-nine-fold gain over the baseline) while reducing overshoot from 76 to 10 percent. The average best areas of $+H_{\mathrm{pair}}$ and the full model are statistically indistinguishable ($33.47\pm 1.09$ against $33.40\pm 0.71$). Two further observations. The area budget bounds over-decoration without eliminating orphan placements, since in all four regimes more than 98 percent of structurally valid samples carry placement bits outside both envelope chains. And the head-to-head run of Table \ref{tab:comparison}, under its own heavier budget, shows the full-model overshoot at $11.6\pm 0.5$ percent, consistent with the regime ordering measured here.
 
\begin{table}[!t]
\centering\footnotesize
\caption{Ablation of the area and pairwise-containment penalties in the sequence-pair formulation over fifteen seeds (the eleven of Section \ref{sec:setup} plus $\{11, 22, 33, 44\}$; 5\,000 reads per seed. Remaining multipliers fixed at their tuned values). Overshoot is the fraction of structurally valid samples whose decorated area exceeds the plate. Paired Wilcoxon against the baseline: geometric feasibility and overshoot $p=6.1\cdot 10^{-5}$ for $+H_{\mathrm{area}}$ and $p=6.5\cdot 10^{-4}$ for $+H_{\mathrm{pair}}$ and full. Best area $p=9.6\cdot 10^{-4}$ for $+H_{\mathrm{pair}}$, $p=1.4\cdot 10^{-3}$ for full, $p=0.18$ for $+H_{\mathrm{area}}$.}
\label{tab:ablation}
\begin{tabular}{lrrr}
\toprule
Regime & Geom.\ feas.\ (\%) & Overshoot (\%) & Best area (avg) \\
\midrule
baseline & $0.62 \pm 0.12$ & $76.00 \pm 0.59$ & $30.07 \pm 1.65$ \\
$+H_{\mathrm{area}}$ & $2.03 \pm 0.26$ & $40.07 \pm 0.44$ & $30.67 \pm 1.70$ \\
$+H_{\mathrm{pair}}$ & $12.83 \pm 0.51$ & $22.42 \pm 0.56$ & $33.47 \pm 1.09$ \\
full & $17.81 \pm 0.33$ & $9.66 \pm 0.46$ & $33.40 \pm 0.71$ \\
\bottomrule
\end{tabular}
\end{table}
 
\subsection{Sensitivity and Calibration Stability}
Perturbing each of the seven multipliers by $\pm 25$ percent in isolation (seed 42, 6\,000 reads $\times$ 1\,500 sweeps) keeps the average best feasible area within $[32,35]$ around the base value of 33 and the geometric feasibility within $[9.9,12.4]$ percent around 10.7 percent, consistent with the multi-seed level of Table \ref{tab:comparison}; no multiplier sits on a knife edge. Calibration is stable in a second sense. The fresh 60-trial TPE study over the documented space of Section \ref{sec:setup} ends with three trials tied at the maximum attainable score, and re-evaluating each on five seeds at the same budget gives best feasible areas of $33.6\pm 1.2$, $32.6\pm 1.4$ and $32.8\pm 1.2$ with geometric feasibility between 7.1 and 7.5 percent and a geometrically infeasible minimum-energy sample in all fifteen runs, statistically indistinguishable from the configuration of Table \ref{tab:comparison}. The search lands reliably on multiplier sets of the same quality, and none of them escapes Proposition \ref{prop:certificate}.
 
\subsection{Hardware: Minor Embedding, QPU Sampling and Hybrid Solves}
The set-packing QUBO minor-embeds on the Advantage\_system4 working graph: 182 logical variables map to 3\,342 physical qubits with maximum chain length 31 and average 18.4 (minorminer is randomised. An earlier exploratory run produced 4\,185 qubits with chains up to 41). The sequence-pair QUBO does not embed: minorminer exhausts its budget without finding an embedding, and in earlier runs aborted at initialisation reporting the source graph as unreasonably large. This failure is structural, not heuristic, as the following bound shows.
 
\begin{proposition}\label{prop:embed}
Any minor embedding of the complete graph $K_{n}$ into a hardware graph of maximum degree $\Delta$ uses at least $n\,\lceil (n-3)/(\Delta-2)\rceil$ physical qubits.
\end{proposition}
\noindent\textit{Proof:} A branch set on $L$ vertices has at most $\Delta L$ incident edge endpoints, spends at least $2(L-1)$ of them on its internal connectivity, and must reach each of the other $n-1$ branch sets through at least one distinct external edge, so $(\Delta-2)L+2\geq n-1$. \hfill$\blacksquare$
 
The area-budget penalty alone induces $K_{545}$ (Table \ref{tab:blocks}). With $\Delta=15$ on the Pegasus topology \cite{boothby2020}, Proposition \ref{prop:embed} demands at least $545\cdot 42=22\,890$ qubits, about four times the roughly 5\,600 of an Advantage chip, with $\Delta=20$ on Zephyr \cite{boothby2021} it demands $545\cdot 31=16\,895$ against roughly 4\,400 on Advantage2. For context, the largest native clique embeddings of these topologies are $K_{12M-10}$ ($K_{182}$ on Pegasus P16) and $K_{16m-8}$ on Zephyr $Z_{m}$ \cite{boothby2020,boothby2021}. The global equality penalties of the coordinate-free model are therefore incompatible with direct execution on any current or announced annealing topology, independently of tuning.
 
The coordinate-based model has its own, softer frontier. Sweeping the plate size at fixed catalogue with the same embedder (180 s, three tries per plate), the $7\times 5$ conflict graph embeds (3\,346 qubits, maximum chain 32), but every larger plate fails, starting already at $8\times 6$ with 287 variables and 13\,941 couplers. The arithmetic is unforgiving: at this conflict density the average chain length is about 18, so a 5\,627-qubit Advantage accommodates on the order of 300 logical variables, and the $8\times 6$ instance sits at that ceiling. The degree-based bound of Proposition \ref{prop:embed} applied per variable requires only about 2\,100 qubits for $8\times 6$, so this failure is one of heuristic embedding and chip capacity rather than provable impossibility, in contrast to the sequence-pair case, but the practical message stands, namely that the linear coupler growth of the set-packing model buys asymptotic headroom on future hardware, not immediate reach on current chips.
 
Direct QPU sampling fails across the entire calibration range, and the two failure modes are instructive. An initial exploratory sweep at chain strengths 0.5 to 4 (1\,000 reads each, on the earlier 4\,185-qubit embedding) broke 90 to 96 percent of chains; majority-vote repair then acts as a decimation that switches most placements off, which yields conflict-free but nearly empty layouts, with best areas of 20, 13, 10 and 18 and conflict-free rates up to 97.2 percent. The calibrated sweep on the fixed 3\,342-qubit embedding tests chain strengths 20, 40, 60 and 80 plus the uniform-torque-compensation value of 448.8, at annealing times of 20, 100 and 200 $\mu$s, again with 1\,000 reads per setting: chain-break fractions fall from 0.70 to 0.008 as the chains stiffen, yet not one of the 7\,000 returned samples is conflict-free, best feasible areas are zero throughout, and minimum energies are positive, between $+87$ and $+2\,013$ against the optimum at $-33.56$. The mechanism is precision rather than chains: with chain couplings at 80 to 449, hardware autoscaling compresses the problem couplers ($|J|\leq 36.2$) and especially the linear terms ($|h|\leq 11.6$) toward the low end of the analog range, and the returned states are dense with overlapping placements. The instance is therefore embeddable but not solvable by direct annealing at any tested setting. On this problem class the quantum-accessible route is hybrid.
 
The Leap hybrid BQM solver closes the loop on both sides of the comparison. On the set-packing model it returns the proven optimum: energy $-33.559=-\mu A^{\star}+\varepsilon n^{\star}$ with $n^{\star}=4$, matched to machine precision, a geometrically feasible sample of area 35, 3.00 s of charge time and 98 ms of QPU access time. On the sequence-pair model it returns the certificate phenomenon verbatim: a single structurally valid sample with decorated area exactly 35 whose energy ties the certificate value of Proposition \ref{prop:certificate} to within $3\cdot 10^{-13}$ and whose decoded layout overflows the plate, in 4.49 s of charge time, so the best geometrically feasible area extracted from the hybrid answer is zero. A solver that returns one lowest-energy sample cannot distinguish the feasible from the infeasible members of the degenerate ground manifold, and here it drew an infeasible one.
 
The constraint-native CQM solver sharpens the same point. The set-packing CQM (35 cover constraints, no penalty weights) is solved outright: 91 feasible samples, best area 35, in 5.3 s of hybrid run time. The sequence-pair CQM carries every expressible constraint of Section \ref{sec:seqpair} natively, 1\,099 in total, with no penalty tuning, no slack bits and no equality cliques in the quantum-visible part. The solver returns 82 constraint-feasible samples whose best objective value is the decorated area 35, and not one of the 82 decodes inside the plate. With the penalty method removed entirely, the failure survives intact, which is Theorem \ref{thm:impossible} observed in its cleanest form: what is missing is not calibration but expressiveness, because the longest-path containment constraint is simply not in the model.
 
\subsection{Simulated Quantum Annealing}
Under OpenJij SQA the set-packing model behaves like an easy sparse problem: all samples conflict-free, optimum 35 found, energy-area correlation $-0.978$, minimum-energy sample feasible, in 78 s of wall time at 1\,000 reads $\times$ 1\,500 sweeps. The sequence-pair model yields no structurally valid sample at all in 2\,438 s, even under the reinforced schedule of Section \ref{sec:setup}. We report this as a genuine negative result rather than tuning it away: hard one-hot permutation structure under a plain transverse field is a known weak spot of SQA without tailored moves, and it compounds the feasibility findings above.
 
\subsection{Wall-Clock Times}
Table \ref{tab:timing} collects the wall-clock times, all measured in the same consolidated session on one workstation. The exact ILP proves optimality in 7.4 ms (median of five runs), which fixes the scale of the comparison: on instances of this size nothing quantum or quantum-inspired competes on time, and the value of the exercise is structural, not chronometric. Between the two QUBO models the asymmetry is itself a result: a sequence-pair SA seed costs $4\,750\pm 52$ s against $24.2\pm 0.6$ s for set-packing, a factor of 196 in wall time, or roughly 74 after normalising by the 2.7-fold larger read-and-sweep budget, in line with the 64-fold coupler ratio of Section \ref{sec:sizes}.
 
\begin{table}[!t]
\centering\footnotesize
\caption{Wall-clock time by QUBO formulation on the reference instance, measured in one consolidated session. SA budgets are 8\,000 reads $\times$ 2\,000 sweeps (sequence-pair) and 4\,000 $\times$ 1\,500 (set-packing). SQA uses 8\,000 against 1\,500 sweeps. The classical ILP baseline (CBC) solves the instance to proven optimality in approximately 7.4 ms (discussed in text).}
\label{tab:timing}
\begin{tabular}{lrr}
\toprule
Time & Sequence-pair & Set-packing \\
\midrule
SA total, 11 seeds (s) & 52\,253 & 267 \\
SA per seed (s) & $4\,750 \pm 52$ & $24.2 \pm 0.6$ \\
SQA, 1\,000 reads (s) & 2\,438 & 78 \\
Leap hybrid BQM, charge time (s) & 4.49 & 3.00 \\
Leap hybrid BQM, QPU access (s) & 0.104 & 0.098 \\
Leap hybrid CQM, run time (s) & 5.20 & 5.25 \\
\bottomrule
\end{tabular}
\end{table}
 
\subsection{Scaling Projection}
\label{sec:scaling}
Fig. \ref{fig:scaling} plots the set-packing variable count $V_{\mathrm{sp}}$ of \eqref{eq:vsp} against the plate area, together with the sequence-pair count of \eqref{eq:vseq} at several fixed slot budgets $N$, each a horizontal line because the coordinate-free model is independent of the plate resolution and size. Three conclusions follow.

First, in the cutting-stock regime proper, where free repetition forces the budget to scale as $N_{\mathrm{bud}}=\lfloor W'H'/a_{\min}\rfloor$, the sequence-pair count $34\,N_{\mathrm{bud}}^{2}$ grows with the square of the plate area against the linear growth of $V_{\mathrm{sp}}\approx P\,W'H'$, so the coordinate-free model loses at every size: 4\,114 against 182 variables already on the reference plate, and $4.6\cdot 10^{7}$ against $3.6\cdot 10^{4}$ at $70\times 50$.

Second, the variable advantage of the coordinate-free model is confined to the opposite regime, a piece budget bounded exogenously by the application rather than by the plate, together with a large reduced plate: at fixed $N$ its count is constant, so the rising $V_{\mathrm{sp}}$ curve crosses it from below and only past that crossing does the sequence-pair model use fewer variables. The crossing sits at plate area $W'H'\approx 3N^{2}$ after the GCD reduction of Lemma \ref{lem:gcd}, and moves to larger plates as $N$ grows: the $N=6$ budget is overtaken just past the $14\times 10$ plate, whereas $N=12$ needs a plate approaching $28\times 20$.

Third, the set-packing maximum degree saturates at 289 from $15\times 11$ onward, because a placement conflicts only with placements overlapping it, a catalogue-bounded neighbourhood. The graph therefore keeps bounded degree and linear coupler growth as the plate grows, although Section \ref{sec:results} shows that chain lengths near 18 already cap heuristic embedding around 300 logical variables on current hardware, while the sequence-pair coupler count remains dominated by equality cliques growing as $\Theta((N^{2}P)^{2})$ and is non-embeddable outright.

Variable counts alone therefore understate the hardware gap, in both directions. The two formulations consequently diverge over time. The set-packing model is capacity and precision-limited today, near the $8\times 6$ plate, but its bounded degree and linear coupler growth make it the formulation that benefits monotonically from larger and denser future topologies. The sequence-pair model faces two obstacles that hardware scaling does not remove: its equality cliques keep it non-embeddable as the topology parameter grows only linearly in native clique size, and the impossibility result of Section \ref{sec:feasibility} denies it a geometrically feasible ground state on any annealer irrespective of size, so its only viable role, now and later, is as a classical or hybrid sample-and-filter heuristic.
 
\begin{figure}[!t]
\centering
\includegraphics[width=7cm]{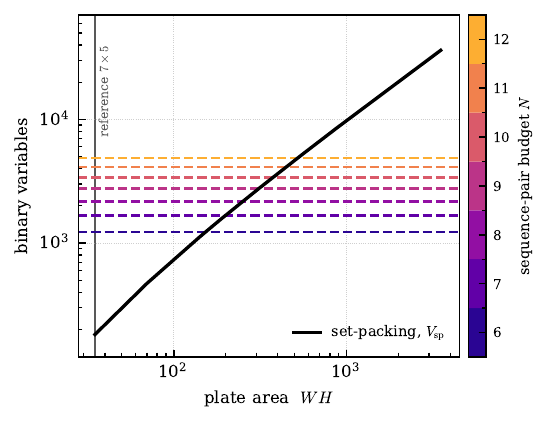}
\caption{Variable-count scaling of the two formulations for the fixed expanded catalogue ($P=11$ orientations, $a_{\min}=3$) across plates with $g=1$, from the closed forms \eqref{eq:vsp} and \eqref{eq:vseq}. The solid black line is the set-packing count $V_{\mathrm{sp}}$, which grows linearly in the plate area $WH$. The dashed lines are the sequence-pair count at fixed slot budgets $N=6,\dots,12$ (value $N^{2}(3P+1)$, with the $O(\log WH)$ slack bits omitted), constant because the coordinate-free model is independent of the plate. For a fixed $N$ the coordinate-free model uses fewer variables only to the right of the point where $V_{\mathrm{sp}}$ crosses its line, that is for plate area $W'H'\gtrsim 3N^{2}$ after GCD reduction, and that crossing moves to larger plates as $N$ grows. The coupler gap is wider still: $V_{\mathrm{sp}}$ keeps catalogue-bounded degree (saturating at $289$ from $15\times 11$ onward) and linear coupler growth, whereas the sequence-pair coupler count is dominated by equality cliques growing as $\Theta((N^{2}P)^{2})$.}
\label{fig:scaling}
\end{figure}
 
\subsection{Multi-Instance Campaign}
\label{sec:campaignres}
Table \ref{tab:campaign} extends the comparison to six instances under the campaign protocol of Section \ref{sec:setup}. Three findings carry over from the reference instance and one is new. First, the set-packing model with the purely analytic multipliers of Proposition \ref{prop:exact}, with no tuning of any kind, reaches the ILP optimum in all 66 runs and returns a geometrically feasible minimum-energy sample in all 66, across reduced plates from $7\times 5$ to $12\times 8$ and four different catalogues. Second, the sequence-pair model, despite a fresh 25-trial calibration per instance, reaches the optimum in 22 of 66 runs, with geometric feasibility between 9.3 and 30.1 percent and a feasible minimum-energy sample in only 3 of 66 runs. Feasibility degrades visibly as the slot budget grows from six to eight (instances C3 and C4). Third, the GCD instance C5 behaves exactly as Lemma \ref{lem:gcd} promises: the raw $14\times 10$ formulation would carry 551 placement variables and 82\,267 conflict couplers, the reduced $7\times 5$ instance carries 182 and 7\,005, and solving it reproduces the reference behaviour, with the physical optimum recovered through the $g^{2}=4$ area scaling of Lemma \ref{lem:gcd} as $4\cdot 35=140$. The new observation is cost: even at the lighter campaign budgets, a sequence-pair seed takes 304 to 2\,899 s against 13 to 45 s for set-packing, so the compact encoding is also the expensive one to sample classically.
 
\begin{table*}[!t]
\centering\footnotesize
\caption{Multi-instance campaign, eleven seeds per model and instance. Set-packing uses the analytic multipliers of Proposition \ref{prop:exact} with no tuning (3\,000 reads $\times$ 1\,200 sweeps). Sequence-pair is recalibrated per instance with 25 TPE trials and sampled at 5\,000 $\times$ 1\,500, with slot budget $N$ set to the ILP piece count plus two (minimum six). $g$ is the divisor of Lemma \ref{lem:gcd}; $V_{\mathrm{sp}}$ and $V_{\mathrm{seq}}$ are the variable counts of the reduced instance that is actually solved, whereas the plate, $A^{\star}$ and the area columns are physical, so for the GCD instance C5 they exceed the reduced run by the factor $g^{2}=4$. C5 is the doubled reference catalogue on a doubled plate ($14\times 10$; 551 variables and 82\,267 couplers before reduction, $A^{\star}=4\cdot 35=140$). Because $N$ and the multipliers differ from the head-to-head protocol, the C1 row is not directly comparable with Table \ref{tab:comparison}. Best area is the per-seed average and is bounded above by $A^{\star}$; per-seed spread is read from the feasibility and at-$A^{\star}$ columns rather than from a symmetric deviation.}
\label{tab:campaign}
\begin{tabular}{llrrrrrrrrrr}
\toprule
 & & & & & & & \multicolumn{2}{c}{Set-packing} & \multicolumn{3}{c}{Sequence-pair} \\
\cmidrule(lr){8-9}\cmidrule(lr){10-12}
Instance & Plate & $g$ & $V_{\mathrm{sp}}$ & $V_{\mathrm{seq}}$ & $A^{\star}$ & $n$ & best area (avg) & at $A^{\star}$ & best area (avg) & feas.\ (\%) & at $A^{\star}$ \\
\midrule
C1 & $7\times 5$ & 1 & 182 & 1\,238 & 35 & 4 & $35.00$ & 11/11 & $34.00$ & 21.1 & 6/11 \\
C2 & $8\times 6$ & 1 & 161 & 1\,023 & 48 & 4 & $48.00$ & 11/11 & $47.18$ & 30.1 & 6/11 \\
C3 & $9\times 6$ & 1 & 205 & 1\,807 & 54 & 6 & $54.00$ & 11/11 & $51.91$ & 10.3 & 2/11 \\
C4 & $10\times 7$ & 1 & 220 & 1\,808 & 70 & 6 & $70.00$ & 11/11 & $67.45$ & 9.3 & 1/11 \\
C5 & $14\times 10$ & 2 & 182 & 1\,238 & 140 & 4 & $140.00$ & 11/11 & $136.00$ & 21.1 & 6/11 \\
C6 & $12\times 8$ & 1 & 342 & 1\,131 & 96 & 4 & $96.00$ & 11/11 & $91.36$ & 13.5 & 1/11 \\
\bottomrule
\end{tabular}
\end{table*}
 
\section{Discussion}
\label{sec:discussion}
 
\subsection{A Practical Selection Rule}
After the GCD reduction of Lemma \ref{lem:gcd}, the recommendation is asymmetric. In the cutting-stock regime proper, free repetition with dense packings, the slot budget must scale as $N\approx W'H'/a_{\min}$ for the coordinate-free model even to represent the optimum, and Fig. \ref{fig:scaling} shows that it then loses on variables at every resolution. The coordinate-based model dominates outright. The coordinate-free model pays off in variable count only when two conditions hold simultaneously: the piece count is bounded a priori by the application (unique items, demand caps, floorplanning-style instances, the setting of \cite{okada2024}), and the reduced plate is fine-grained, approximately $3N^{2}<W'H'$. Even then, variable economy is not the binding constraint on annealers: the global equality penalties make the sequence-pair interaction graph clique-bearing and provably non-embeddable (Proposition \ref{prop:embed}), while the set-packing graph keeps catalogue-bounded degree, linear coupler growth and, within the capacity limits measured in Section \ref{sec:results}, an actual embedding. On classical or hybrid samplers, where both models fit, the coordinate-free model must be operated as a sample-and-filter heuristic, with the feasibility and correlation costs quantified in Section \ref{sec:results} and the ground-state caveat of Proposition \ref{prop:certificate} kept in view whenever a solver returns a single lowest-energy sample. The constraint-native CQM results of Section \ref{sec:results} show that replacing penalties by native constraints does not lift this caveat.

Table \ref{tab:selection} summarises the selection rule.

\begin{table}[!t]
\centering\footnotesize
\caption{Decision rule for choosing between the two formulations, after the GCD reduction of Lemma \ref{lem:gcd}. $N$ is the slot budget and $W'H'$ the reduced plate area. Three of the four operating points select the coordinate-based model.}
\label{tab:selection}
\begin{tabular}{@{}p{2.45cm}p{1.35cm}p{3.05cm}@{}}
\toprule
Operating point & Choose & Reason \\
\midrule
Free repetition, dense packing (any resolution) & Set-packing & Fewer variables and couplers after GCD; exact and embeddable \\
\addlinespace
Bounded pieces, coarse plate ($3N^{2}\!\geq\! W'H'$) & Set-packing & Sequence-pair variable advantage not yet active \\
\addlinespace
Bounded pieces, fine plate, annealing hardware & Set-packing & Sequence-pair non-embeddable (cliques, Prop.~\ref{prop:embed}); set-packing fits to $\sim\!8\!\times\!6$ \\
\addlinespace
Bounded pieces, fine plate, classical or hybrid, heuristic OK & Sequence-pair & Only regime where compact variables pay; heuristic, low feasibility (Prop.~\ref{prop:certificate}) \\
\bottomrule
\end{tabular}
\end{table}
 
\subsection{Limitations}
The hardware study uses a single QPU generation (Advantage\_system4), while Zephyr-topology chips enter only through the bound of Proposition \ref{prop:embed}, and the embedding sweep uses minorminer with default settings, so clique-aware or specialised embedders might push the capacity frontier somewhat beyond the $8\times 6$ plate. The coordinate-free multipliers are calibrated per instance, and Remark \ref{rem:soft} explains why this cannot be fully avoided. SQA schedules were not tuned per formulation. All campaign instances admit perfect packings. Instances with $A^{\star}<W'H'$, where the area budget cannot be saturated, are a natural extension. Finally, the ablation keeps its own fifteen-seed protocol and predates the consolidated session. In contrast, every other number in Section \ref{sec:results} comes from that single session.
 
\subsection{Extensions}
Three directions follow naturally. Radix-2 coordinate encodings occupy a genuine middle point of the axis, with $\Theta(N\log W'H')$ variables, but non-overlap then becomes a high-degree product of coordinate bits, so the faithfulness question of Section \ref{sec:feasibility} reappears after quadratization. Lazy constraint generation sidesteps the $10^{8}$-term blow-up of Section \ref{sec:feasibility}: sample, decode, and quadratize only the violated chain monomials, at the price of multiple solver rounds. Finally, constraint-native hybrid solvers have precedent on real packing workloads \cite{romero2023}, and Section \ref{sec:results} confirms both halves of the expectation: the CQM reformulation removes the equality cliques and solves the set-packing model outright, yet cannot rescue the sequence-pair model, because containment itself is what the representation lacks. Demand-bounded and guillotine variants, which would connect the present study with \cite{arai2021}, and embedding strategies specialised to dense local conflict graphs are the remaining directions.
 
\section{Conclusion}
\label{sec:conclusion}
We compared two QUBO encodings of the fixed-plate 2D Cutting Stock Problem with free repetition that represent fundamentally different encoding paradigms. The coordinate-based set-packing model is exact, of catalogue-bounded degree and linear coupler growth, and minor-embeddable within the capacity of current chips. It was solved to proven optimality by simulated annealing, simulated quantum annealing, and both Leap hybrid solvers on the reference instance, reaching the exact optimum in all 66 campaign runs with purely analytic multipliers. Direct QPU sampling nevertheless fails at every tested calibration, and heuristic embedding caps out near 300 logical variables, exposing two honest limits of present hardware.

The coordinate-free sequence-pair model is compact in variables on fine grids. However, no encoding of its kind with bounded interaction degree and penalties that vanish on feasible layouts can keep its ground state geometrically feasible. We proved this, exhibited an explicit energy-degenerate infeasible certificate of the tuned Hamiltonian, and watched the hybrid solver return it. The certificate energy was reproduced to thirteen decimal places with the decoded layout overflowing the plate, a behaviour that survives unchanged when every expressible constraint is imposed natively in a CQM. This fundamental lack of expressiveness is compounded by a provable minor-embedding separation driven by its global equality cliques.

The resulting selection rule is asymmetric. With free repetition, the coordinate-based model dominates at every resolution after GCD reduction. However, the coordinate-free model becomes a highly valuable alternative when scaling variable counts prohibit exact encodings. Deployed as a sample-and-filter heuristic, it successfully discovers optimal maximum-area layouts, offering a robust practical path when the exact formulation exhausts hardware capacity. We make no claim of quantum speedup. Our contribution is a structural map of the design space that tells a practitioner which encoding to trust, when, and why.

\end{document}